\newif\ifplain
\newif\ifieee
\newif\iflncs

\plaintrue
\ieeefalse
\lncsfalse

\ifplain
\documentclass[a4paper]{article}
\newenvironment{stdenum}{\begin{enumerate}}{\end{enumerate}}
\fi

\ifieee
\documentclass[10pt, conference]{IEEEtran}
\usepackage{paralist}
\newenvironment{stdenum}{\begin{enumerate}}{\end{enumerate}}
\usepackage{flushend} 
\fi

\iflncs
\documentclass[a4paper,envcountsame]{llncs}
\usepackage{times}
\usepackage{paralist}
\newenvironment{stdenum}{\begin{inparaenum}[(\itshape i\upshape)]}{\end{inparaenum}}
\fi

\ifplain
\newcommand{\exOnly}[1]{#1}
\newcommand{\shOnly}[1]{}
\else
\newcommand{\exOnly}[1]{}
\newcommand{\shOnly}[1]{#1}
\fi

\usepackage{amsmath,amssymb}
\iflncs
\usepackage{ntheorem}
\else
\usepackage{amsthm}
\fi
\usepackage{dsfont}
\usepackage{hyperref}
\usepackage{xspace}

\usepackage{counterMachine}
\newcommand{\fakepar}[1]{\vskip 2pt\textbf{#1}}

\usepackage{triotex}

\newcommand{\CntMTL}[2]{\MTLoperator{\mathsf{K}}{#1}{#2}{}}
\newcommand{\nextMTL}[2]{\MTLoperator{\bigcirc}{#1}{}{#2}}

\usepackage{tikz}
\usetikzlibrary{calc,patterns,decorations,decorations.pathmorphing,decorations.pathreplacing,decorations.markings,shapes,arrows,positioning,petri,mindmap,fit,backgrounds,chains}
\usetikzlibrary{intersections}
\usetikzlibrary{automata,shapes.symbols}
\usetikzlibrary{external}
\tikzset{external/only named=true}

\newcommand{\etap}[1]{\ifthenelse{\equal{#1}{1}} {\eta^{\ref{eq:normalform}}} {\ifthenelse{\equal{#1}{2}} {\eta^{\ref{eq:normalform-G}}} {\eta^{\ref{eq:normalform-X}}}}}

\DeclareRobustCommand{\form}[1]%
{\ensuremath{%
	 \ifthenelse{ \not \equal{#1}{} } {\mathtt{L}\!\left(#1\right)} {\mathtt{L}} }
}%

\DeclareRobustCommand{\lang}[1]%
{\ensuremath{%
	 \ifthenelse{ \not \equal{#1}{} } {\mathbf{L}\!\left(#1\right)} {\mathbf{L}} }
}%

\DeclareRobustCommand{\formk}[2]{\ifthenelse{ \not \equal{#1}{} \AND \not \equal{#2}{} }{\form{\mathsf{U}^{#1}, \mathsf{X}^{#2}}}{\ifthenelse{ \not \equal{#1}{} }{\form{\mathsf{U}^{#1}, \mathsf{X}}}{\ifthenelse{ \not \equal{#2}{} }{\form{\mathsf{U}, \mathsf{X}^{#2}}}{\form{\mathsf{U}, \mathsf{X}}}}}}

\DeclareRobustCommand{\langk}[2]{\ifthenelse{ \not \equal{#1}{} \AND \not \equal{#2}{} }{\lang{\mathsf{U}^{#1}, \mathsf{X}^{#2}}}{\ifthenelse{ \not \equal{#1}{} }{\lang{\mathsf{U}^{#1}, \mathsf{X}}}{\ifthenelse{ \not \equal{#2}{} }{\lang{\mathsf{U}, \mathsf{X}^{#2}}}{\lang{\mathsf{U}, \mathsf{X}}}}}}

\newcommand{\logtrue}{\top}
\newcommand{\logfalse}{\bot}
\newcommand{\limpl}{\Rightarrow}
\newcommand{\liff}{\Leftrightarrow}

\newcommand{\Pcal}{\mathcal{P}}

\newcommand{\Fcal}{\mathcal{F}}

\newcommand{\Orm}{\mathrm{O}}
\newcommand{\orm}{\mathrm{o}}

\newcommand{\allmodels}[2]{\ensuremath{\mathcal{B}{#1}\ifthenelse{\not \equal{#2}{}}{(#2)}{}}}

\newcommand{\alltw}[1]{\allmodels{\timedomain}{#1}}

\iflncs
\theorembodyfont{\upshape}
\newtheorem{Definition}{Definition}
\else
\theoremstyle{plain}
\newtheorem{theorem}{Theorem}

\newtheorem{lemma}[theorem]{Lemma}
\newtheorem{corollary}[theorem]{Corollary}
\theoremstyle{definition}
\newtheorem{definition}[theorem]{Definition}

\newtheorem{remark}[theorem]{Remark}
\fi

\newcommand{\bw}[3]{\ensuremath{\mathcal{B}{#1}\ifthenelse{\not \equal{#2}{}}{[#2]}{}}\ifthenelse{\not \equal{#3}{}}{(#3)}{}}
\newcommand{\bwc}[2]{\bw{\reals_{\geq 0}}{#1}{#2}}

\newcommand{\bwt}[2]{\bw{\timedomain}{#1}{#2}}

\newcommand{\BV}[2]{\ensuremath{\mathit{BV}_{#1}\ifthenelse{\not \equal{#2}{} }{(#2)}{}}}
\newcommand{\cBV}[2]{\ensuremath{\mathit{\overline{BV}}_{#1}\ifthenelse{\not \equal{#2}{} }{(#2)}{}}}
\newcommand{\BVd}[1]{\BV{\naturals}{#1}}
\newcommand{\BVc}[1]{\BV{\reals_{\geq 0}}{#1}}
\newcommand{\BVt}[1]{\BV{\timedomain}{#1}}
\newcommand{\cBVd}[1]{\cBV{\naturals}{#1}}
\newcommand{\cBVc}[1]{\cBV{\reals_{\geq 0}}{#1}}
\newcommand{\cBVt}[1]{\cBV{\timedomain}{#1}}

\newcommand{\RE}{\languageclass{\textup{RE}}{}{}{}{}}
\newcommand{\cRE}{\languageclass{\textup{coRE}}{}{}{}{}}
\newcommand{\expspace}{\languageclass{\textup{EXPSPACE}}{}{}{}{}}
\newcommand{\exptime}{\languageclass{\textup{EXP}}{}{}{}{}}
\newcommand{\np}{\languageclass{\textup{NP}}{}{}{}{}}

\newcommand{\frag}[4]{\ensuremath{\mathrm{#1}\ifthenelse{\not \equal{#2}{} }{^{#2}}{}\ifthenelse{\not \equal{#3}{} }{_{#3}}{}\ifthenelse{\not \equal{#4}{} }{\left[#4\right]}{}}}
\newcommand{\mtlFX}[1]{\frag{\Fcal}{+}{\diamondMTL{}{}, \nextMTL{}{}}{}}
\newcommand{\mtlGX}[1]{\frag{\Fcal}{+}{\boxMTL{}{}, \nextMTL{}{}}{}}
\newcommand{\ltlFX}[1]{\frag{\mathcal{L}}{+}{\Fmtl{}{}, \Xmtl{}{}}{}}

\newcommand{\auntilMTL}[2]{\MTLoperator{\widehat{\mathsf{U}}}{#1}{}{#2}}
\newcommand{\pwm}{\models_{\textsf{p}}}
\newcommand{\cm}{\models_{\textsf{c}}}

\newcommand{\noop}{\epsilon}
\newcommand{\ob}{\overline{b}}

\ifplain
\date{10 June 2013\footnote{Document last updated on 4 July 2014.}}
\title{\textsc{Bounded Variability of \\Metric Temporal Logic}}
\author{Carlo A. Furia \\
Department of Computer Science, ETH Zurich, Switzerland\\
\url{bugcounting.net} \\
Paola Spoletini \\
Universit{\`a} degli Studi dell'Insubria, Italy \\
\url{paola.spoletini@uninsubria.it}
}
\fi
\ifieee
\title{Bounded Variability of Metric Temporal Logic}
\author{%
  \IEEEauthorblockN{Carlo A.\ Furia}
  \IEEEauthorblockA{Department of Computer Science,  ETH Zurich, Switzerland\\
  \url{caf@inf.ethz.ch}
  }
\and
  \IEEEauthorblockN{Paola Spoletini}
  \IEEEauthorblockA{Universit{\`a} degli Studi dell'Insubria, Italy\\
  \url{paola.spoletini@uninsubria.it}
  }
}
\fi
\iflncs
\title{Bounded Variability of Metric Temporal Logic}
\author{Carlo A.\ Furia\inst{1} \and Paola Spoletini\inst{2}}
\institute{%
Department of Computer Science, ETH Zurich, Switzerland $\quad$ \email{caf@inf.ethz.ch} 
\and
Universit{\`a} degli Studi dell'Insubria, Italy $\quad$ \email{paola.spoletini@uninsubria.it}
}
\fi

\begin{document}

\maketitle

\begin{abstract}
Previous work has shown that reasoning with real-time temporal logics is often simpler when restricted to models with \emph{bounded variability}---where no more than $v$ events may occur every $V$ time units, for given $v, V$.
When reasoning about formulas with \emph{intrinsic} bounded variability, one can employ the simpler techniques that rely on bounded variability, without any loss of generality.
What is then the complexity of algorithmically \emph{deciding} which formulas have intrinsic bounded variability?

In this paper, we study the problem with reference to Metric Temporal Logic (MTL).
We prove that deciding bounded variability of MTL formulas is undecidable over dense-time models, but with a undecidability degree lower than generic dense-time MTL satisfiability.
Over discrete-time models, instead, deciding MTL bounded variability has the same exponential-space complexity as satisfiability.
To complement these negative results, we also briefly discuss small fragments of MTL that are more amenable to reasoning about bounded variability. 
\end{abstract}

\section{The Benefits of Bounding Variability} \label{sec:intro}
In yet another instance of the principle that ``there ain't no such thing as a free lunch'', expressiveness of formal languages comes with a significant cost to pay in terms of complexity---and possibly undecidability---of algorithmic analysis.
The trade-off between expressiveness and complexity is particularly critical for the real-time temporal logics, which dwell on the border of intractability.
A chief research challenge is, therefore, identifying expressive temporal logic fragments without letting the ``dark side'' of undecidability~\cite{DBLP:conf/time/BresolinMGMS11} prevail and abate practical usability.

Previous work by us~\cite{FR08-FORMATS08,FS11-TIME11} and others~\cite{Wil94,DBLP:journals/fac/Franzle04} has shown that the notion of \emph{bounded variability} can help tame the complexity of real-time logics while still retaining a reasonable expressive power.
A model has variability bounded by $v/V$ if there are at most $v$ events every $V$ time units.
Consider a temporal logic formula $\phi$: deciding whether $\phi$ has a model with variability bounded by some $v/V$ is typically simpler than the more general problem of deciding whether $\phi$ has a model of \emph{any} (possibly unbounded) variability (see Section~\ref{sec:related} for examples).
To close the gap between decidability over bounded variably models and general models, we should be able to determine if $\phi$ \emph{only} has models with bounded variability.
When this is the case, we lift the notion of bounded variability from models to formulas and say that $\phi$ has bounded variability.
For formulas with bounded variability, we can apply the simpler algorithms that only consider bounded variably models, without losing generality in the analysis.

As a simple concrete example, if $\phi$ is the specification of a square wave of period $10$ and duty cycle 30\% %
(\begin{tikzpicture}[x=3pt,y=7pt,thick,xscale=1,yscale=1,baseline={(B.base)}]
\draw (-2,0) -- (0,0) -- (0,1) -- (3,1) -- (3,0) -- (10,0) -- (10,1) -- (13,1) -- (13,0) -- (15,0);
\draw [dotted] (15,0) -- (17,0);
\draw [dotted] (-2,0) -- (-4,0);
\draw [thin,<->,blue] (0,-0.5) -- (10,-0.5);
\node at (-1,-0.5) {\scriptsize 0};
\node at (11,-0.5) {\scriptsize 10};
\node (B) at (10,-0.5) {};
\end{tikzpicture}), 
there are at most three transition events every 10 time units.
Thus, all models of $\phi$ have variability bounded by $3/10$, and we can leverage this fact to simplify the algorithmic analysis of $\phi$.

This paper targets the bounded variability of formulas written in Metric Temporal Logic (MTL)~\cite{Koy90}, a popular linear-time temporal logic which extends LTL~\cite{Eme90} with metric constraints and can be interpreted over both dense and discrete time domains.
We study the complexity of the problem of determining if a generic MTL formula $\phi$ has bounded variability.

The bulk of the results is bad news: over dense-time models, deciding whether an MTL formula has bounded variability is undecidable; over discrete-time models, it is decidable, but with the same complexity as deciding validity\footnote{Since MTL formulas are obviously closed under negation, validity and satisfiability are dual problems with the same complexity. Therefore, we indifferently use either term with reference to complexity.} in general.
These results are major hurdles to pursuing the idea of identifying formulas with bounded variability and then using the simpler algorithms for satisfiability on them: the complexity of the first step dominates, and nullifies the benefits of using the simplified algorithms for satisfiability under bounded variability.

As we show in Section~\ref{sec:compl-bound-vari-ct} using reductions from undecidable problems of nondeterministic counter machines, the undecidability degree of deciding bounded variability over dense time is still lower than that of deciding validity: the former occupies the first two levels of the arithmetical hierarchy, whereas the latter belongs to $\Sigma_1^1$, the second level of the analytical hierarchy.
In contrast, deciding bounded variability over discrete time is \expspace-complete (see Section~\ref{sec:compl-bound-vari-dt}), the very same complexity as deciding validity; but the discreteness of the time domain entails that every formula has bounded variability for $v = V$.

While these results imply strong limits to reasoning about bounded variability in general, Section~\ref{sec:bound-vari-simple} suggests simpler cases where this may still be possible.
If we identify MTL fragments that are sufficiently expressive to encode the requirement of bounded variability, yet have low complexity, we can try to establish bounded variability in special cases by considering subformulas of generic MTL formulas.
We briefly illustrate two fragments, one for discrete- and one for dense-time models, that meet these requirements.

\ifplain
\subsection{Related Work} \label{sec:related}
\else
\fakepar{Related Work.} \label{sec:related}
\fi
Originally introduced by Koymans~\cite{Koy90} as a first-order real-time logic, MTL has become widespread in the propositional version popularized by Alur and Henzinger~\cite{AH93}.
Their seminal work has also studied its complexity over dense and discrete time~\cite{AH93,AH94}, as well as interesting decidable fragments for dense time~\cite{AFH96}.
While their work basically settled the problems for discrete time, follow-up work by other authors has extended and refined the picture for dense time, such as by studying expressive completeness~\cite{HR04,HOW13}, simplifying decision procedures~\cite{MNP05}, or identifying expressive decidable fragments~\cite{DBLP:conf/lics/BouyerMOW07,DBLP:journals/lmcs/OuaknineW07,DBLP:conf/formats/OuaknineW08}.

Bounded variability is a natural semantic restriction over dense time, which has been applied to various formalisms including timed automata~\cite{Wil94}, duration calculus~\cite{DBLP:journals/fac/Franzle04}, and, in our previous work, MTL~\cite{FR08-FORMATS08}.
Recently, we also applied it to LTL over discrete time; while it is obvious that every discrete-time model has bounded variability (given by the fixed duration associated with one discrete time unit), in~\cite{FS11-TIME11,FS12-TIME12} we showed how the LTL validity problem can be simplified under the assumption that only $v < V$ change events happen every $V$ discrete time steps.
Therefore, bounded variability can be a simplifying assumption also for discrete time.

The undecidability results of Section~\ref{sec:compl-bound-vari-ct} use reductions from undecidable problems of nondeterministic $n$-counter machines, which we introduce in Section~\ref{sec:counter-machines}.
These are a kind of Minsky's counter machines~\cite{Min67}; their connection with MTL was first exploited by Alur and Henzinger~\cite{AH93}.

Section~\ref{sec:bound-vari-simple} discusses MTL fragments with lower complexity.
Over discrete time, these fragments can be derived from similarly low-complexity fragment of LTL~\cite{Eme90}, which have been extensively studied by several authors~\cite{DBLP:journals/jacm/SistlaC85,DBLP:journals/iandc/DemriS02,LWW07}.


\section{Timed Words and Variability} \label{sec:words}
We denote a generic time domain by $\timedomain$.
In the paper, $\timedomain$ is either the \emph{discrete} set of the nonnegative integers $\naturals$, or the \emph{dense} (and \emph{continuous}) set of the nonnegative reals $\reals_{\geq 0}$.

An \emph{interval} is a convex subset of the time domain, represented by a pair $\langle a, b \rangle$, where $\langle$ and $\rangle$ are square or round brackets to respectively denote inclusion or exclusion of the endpoint.
We use the pseudo-arithmetic expressions $> s$, $\geq s$, $< s$, $\leq s$, and $= s$ as abbreviations for the intervals $(s, \infty)$, $[s, \infty)$, $[0, s)$, $[0, s)$ and $[s,s]$.
We assume a binary encoding of constants in the time domain unless explicitly stated otherwise.

Given a time domain $\timedomain$ and a finite alphabet set $\Pcal$ of atomic propositions, a \emph{timed word} over $\timedomain$ is a countably infinite sequence of pairs $\omega = (\sigma_0, t_0)\,(\sigma_1, t_1)\ifieee\else\linebreak\fi(\sigma_2, t_2) \cdots$ such that:
\begin{stdenum}
\item 
Each integer $k \geq 0$ denotes a \emph{position} in a timed word;
\item 
For each $k$, $\sigma_k$ is a (nonempty, w.l.o.g.) subset of $\Pcal$ denoting the propositions holding at position $k$; and $t_k \in \timedomain$ is a timestamp denoting the time of the occurrence at position $k$;
\item
The timestamps are strictly monotonic, that is $t_h > t_k$ iff $h > k$, and diverging, that is for all $t \in \timedomain$ there exists $k$ such that $t_k > t$ (divergence is subsumed by monotonicity in discrete time).
\end{stdenum}
We also conventionally assume that $t_0 = 0$.
The set of all timed words over $\timedomain$ is denoted by $\alltw{}$.

A timed word $\omega$ has \emph{variability bounded by $v/V$}, for $V \in \timedomain$ and $v \in \naturals$, iff it has no more than $v$ positions within any closed time interval of length $V$: for all $k\in \naturals$, $t_{k+v}-t_{k} > V$.
The set of all timed words over $\timedomain$ with variability bounded by $v/V$ is denoted by $\bwt{v/V}{}$.

\section{MTL: Metric Temporal Logic} \label{sec:mtl}
We present the syntax and semantics of propositional MTL and recall some fundamental facts about its complexity. 

\fakepar{Syntax.}
MTL formulas are defined by the grammar
\iflncs
\else
\begin{equation*}
\phi \quad::=\quad
  \logtrue
  \,\mid\, p
  \,\mid\, \neg \phi_1 
  \,\mid\, \phi_1 \land \phi_2 
  \,\mid\, \untilMTL{J}{\phi_1, \phi_2}
,
\end{equation*}
\fi
where $p$ ranges over the alphabet $\Pcal$, and $J$ is an interval of the time domain with integer endpoints.
We assume the standard definitions for \emph{false}: $\logfalse$, and for the derived Boolean connectives: $\lor$, $\limpl$, and $\liff$.
The symbol $\alpha$ abbreviates the formula $\bigvee_{p \in \Pcal} p$, which holds iff some proposition holds.
We introduce the derived temporal operators \emph{eventually}: $\diamondMTL{J}{\phi} = \untilMTL{J}{\logtrue, \phi}$; \emph{globally} (also, \emph{always}): $\boxMTL{J}{\phi} = \neg \diamondMTL{J}{\neg \phi}$; \emph{action until}: $\auntilMTL{J}{\phi_1, \phi_2} = \untilMTL{J}{\alpha \limpl \phi_1, \phi_2}$; and \emph{next}: $\nextMTL{J}{\phi} = \auntilMTL{J}{\logfalse, \phi}$.
Operator precedence is: $\neg$ has the highest precedence, then $\land$, then $\lor$, then $\limpl$, then all temporal operators, and finally $\liff$.
We may omit the parentheses around arguments when unambiguous, and drop intervals $[0, \infty)$.

\fakepar{Semantics.}
Given a timed word $\omega = (\sigma_0, t_0)\,(\sigma_1, t_1)\cdots$ and a position $k \in \naturals$, the \emph{pointwise} satisfaction relation $\pwm$ for an MTL formula $\phi$ is inductively defined as follows:

\ifieee\setlength{\tabcolsep}{5pt}\else\begin{center}\fi
\begin{tabular}{lcl}
$\omega, k \pwm \logtrue$$;$\\
$\omega, k \pwm p$  & iff & 
  $p \in \sigma_k$ $\,;$ \\
$\omega, k \pwm \neg \phi_1$ & iff &
  $\omega, k \not\pwm \phi_1$$;$ \\
$\omega, k \pwm \phi_1 \land \phi_2$ & iff &
  $\omega, k \pwm \phi_1$ and $\omega, k \pwm \phi_2$$;$\\
\ifieee
$\omega, k \pwm \untilMTL{J}{\phi_1, \phi_2}$ & iff & 
  there exists $h > k$ such that: \\
  && $t_h - t_k \in J$, $\omega, h \pwm \phi_2$, and, \\
  && for all $k < x < h$, $\omega, x \pwm \phi_1$$;$ \\
\else
$\omega, k \pwm \untilMTL{J}{\phi_1, \phi_2}$ & iff & 
  there exists $h > k$ such that: $t_h - t_k \in J$, \\
  && $\omega, h \pwm \phi_2$, and, for all $k < x < h$, $\omega, x \pwm \phi_1$$;$ \\
\fi
$\omega \pwm \phi$ & iff & $\omega, 0 \pwm \phi$$\,.$
\end{tabular}
\ifieee\else\end{center}\fi

\noindent
The semantics of $\untilMTL{}{}$ and $\auntilMTL{}{}$ coincide under the pointwise semantics, since formulas are only evaluated at positions, where $\alpha$ invariably holds.
The (derived) semantics of next is: $\omega, k \pwm \nextMTL{J}{\phi_1}$ iff $t_{k+1} - t_k \in J$ and $\omega, k+1 \pwm \phi_1$; that is, the next position has timestamp in $J$ relative to the current one, and $\phi_1$ holds there.

Given a timed word $\omega$ as above and a time instant $t \in \timedomain$, the \emph{continuous} satisfaction relation $\cm$ for an MTL formula $\phi$ is inductively defined as follows
\ifplain\else(we only list the cases that differ from the pointwise semantics)\fi:

\ifieee
\setlength{\tabcolsep}{6pt}
\begin{tabular}{lcl}
$\omega, t \cm p$  & iff & 
  there exists $k \in \naturals$ such that: \\
  && $t_k = t$ and $p \in \sigma_k$$;$ \\
$\omega, t \cm \untilMTL{J}{\phi_1, \phi_2}$ & iff & 
  there exists $u > t$ such that: \\
&& $u - t \in J$, $\omega, u \cm \phi_2$, and, \\
&& for all $t < v < u$, $\omega, v \cm \phi_1$$.$ \\
\end{tabular}
\else
\begin{center}
\begin{tabular}{lcl}
\ifplain
$\omega, t \cm \logtrue$$;$\\
\fi
$\omega, t \cm p$  & iff & 
  there exists $k \in \naturals$ such that: 
  $t_k = t$ and $p \in \sigma_k$$;$ \\
\ifplain
$\omega, t \cm \neg \phi_1$ & iff &
  $\omega, t \not\cm \phi_1$$;$ \\
\fi
\ifplain
$\omega, t \cm \phi_1 \land \phi_2$ & iff &
  $\omega, t \cm \phi_1$ and $\omega, t \cm \phi_2$$;$\\
\fi
$\omega, t \cm \untilMTL{J}{\phi_1, \phi_2}$ & iff & 
  there exists $u > t$ such that: 
$u - t \in J$,  \\
&&$\omega, u \cm \phi_2$, and, for all $t < v < u$, $\omega, v \cm \phi_1$$;$ \\
\ifplain
$\omega \cm \phi$ & iff & $\omega, 0 \cm \phi$$\,.$
\fi
\end{tabular}
\end{center}
\fi

\noindent
Over dense time, the continuous semantics generalizes the pointwise semantics in the sense that the former is strictly more expressive~\cite{DSouzaP07,DBLP:conf/lics/BouyerMOW07}.
The semantics of next under continuous semantics is, however, analogous to that over the pointwise semantics, thanks to the usage of $\auntilMTL{}{}$ in its definition.

\begin{remark}\label{rem:semantics-used}
In the following, we assume the pointwise semantics over discrete time $\naturals$, and the continuous semantics over dense time $\reals_{\geq 0}$.
Our results for dense time are also transferable to the pointwise semantics \emph{mutatis mutandis}, provided \emph{past operators} are available (see Section~\ref{sec:other-semantics}).
\end{remark}

\fakepar{Complexity: general models.}
Satisfiability of MTL formulas is highly undecidable over dense time, where it is $\Sigma_1^1$-hard~\cite{AH93}.
It is instead decidable over discrete time, with an $\expspace$-complete decidability problem~\cite{AH93} (which translates to doubly-exponential deterministic time).
Over discrete time, the high complexity is essentially due to the succinctness of the binary encoding (the expressiveness is the same as LTL).

\fakepar{Complexity: bounded models.}
Bounded variability is a semantic restriction that reduces the complexity of MTL.
In fact, we proved that satisfiability of MTL over dense-time models with variability bounded by $v/V$, for any given $v/V$, is \expspace-complete~\cite{FR08-FORMATS08}, matching the complexity of MTL over discrete time, as well as that of other decidable dense-time logics~\cite{AFH96,HR04}.
The following is a corollary of our previous results, which we use in this paper.
\begin{corollary}\label{prop:MTL-bounded-sat}
For any $v$, $v'$, and $V$, it is decidable whether an MTL formula has some model over $\reals_{\geq 0}$ with variability bounded by $v'/V$ but not by $v/V$.
\end{corollary}
\begin{proof}
We showed that the MTL satisfiability problem over \bwc{v/V}{} is decidable for generic $v/V$~\cite[Corollary~1]{FR08-FORMATS08}; and that we can encode in MTL the bounded variability constrain as well as its complement~\cite[Section~4.3]{FR08-FORMATS08}.
\iflncs\else\qedhere\fi\end{proof}

In recent work~\cite{FS11-TIME11,FS12-TIME12}, we showed how the notion of bounded variability can reduce the complexity of MTL over discrete time as well.
While bounded variability does not affect the exponential-space worst-case complexity, since discrete-time models have inherently bounded variability, it can reduce the complexity in practice.
Precisely, when studying the variability of an arbitrary LTL formula $\phi$ over behaviors with variability bounded by any given $v/V$, with $v < V$, we can consider a simplified $\phi'$ whose size depends on $v$ but not on the distances encoded in $\phi$ through next operators.
While the results of \cite{FS11-TIME11,FS12-TIME12} target LTL, it is clear that they carry over to MTL over discrete time. 

\section{Counter Machines} \label{sec:counter-machines}
Counter machines~\cite{Min67} are powerful computational devices, widely used in formal language theory.
We use a nondeterministic version of counter machines, and derive some complexity results which we use in the remainder.

\iflncs\begin{Definition}\else\begin{definition}\fi
An $n$-counter machine executes programs consisting of a finite list of instructions with labels $\ell_0, \ell_1, \ldots$ and operating on $n$ integer counter variables $v_0, \ldots, v_{n-1}$.
An instruction is one of the following:

\begin{center}
\setlength{\tabcolsep}{7pt}
\begin{tabular}{l l}
\lstinline|halt| & terminate computation \\
\lstinline|if $v_k > 0$ goto $\ell_i$, $\ell_j$| & conditional branch \\
\lstinline|inc $v_k$| &  increment counter \\
\lstinline|dec $v_k$| &  decrement counter
\end{tabular}
\end{center}

\noindent
where the conditional branch consists in jumping to $\ell_i$ or $\ell_j$ nondeterministically if counter $v_k$ is non-zero; and decrementing a counter with zero value is undefined.
Computations start at location $\ell_0$ with all counters equal to zero and proceed according to the obvious semantics of instructions.
Without loss of generality, assume that instruction \lstinline|halt| occurs exactly once and that the last instruction in the list is either \lstinline|halt| or a branch.
\iflncs\end{Definition}\else\end{definition}\fi

For $n$-counter machines, with $n \geq 2$, the halting problem (deciding whether the location with \lstinline|halt| is visited in some computation) is $\Sigma_1^0$-complete (\RE-complete: undecidable but semidecidable); 
the non-halting problem (deciding whether some computation does not halt) is $\Sigma_2^0$-complete;
the recurring computation problem (deciding whether location $\ell_0$ is visited infinitely often in some computation) is $\Sigma_1^1$-hard~\cite{AH94}.\exOnly{\footnote{\cite{AH94} discusses 2-counter machines, but the generalization to $n$-counter machines is immediate. The other complexities follow from reduction of the same problems for Turing machines.}}\shOnly{\footnote{These complexities easily follow from reduction of the same problems for Turing machines.}}

\subsection{Bounded and Unbounded Counters}\label{sec:bound-unbo-count}
Consider the following decision problems for $n$-counter machines:
\ifieee
\begin{description}[\IEEEsetlabelwidth{unbounded counter: }]
\else
\begin{description}
\fi
\item[\textbf{bounded counter}:] given an integer $\beta$, decide whether $v_0$ overflows $\beta$ in some computation;
\item[\textbf{finite counter}:] decide whether there exists $\beta$ such that $v_0 \leq \beta$ in all computations;
\exOnly{\item[\textbf{unbounded counter}:] decide whether $v_0$ is incremented infinitely often in some computation.}
\end{description}

\begin{theorem}\label{th:counter-problems}
The bounded counter problem is $\Sigma_1^0$-complete;
the finite counter problem is $\Sigma_2^0$-complete\exOnly{;
the unbounded counter problem is $\Sigma_1^1$-hard}.
\end{theorem}

\iflncs
\begin{proof}
See Appendix~\ref{app:bounded-unbounded-counter-problems-proof}.
\iflncs\qed\else\qedhere\fi\end{proof}
\else
\begin{proof}\iflncs[of Theorem~\ref{th:counter-problems}]$\:$\fi 
We prove hardness by reduction from, respectively, the halting\exOnly{,} \shOnly{and} non-halting\exOnly{, and recurring computation} problems of $n$-counter machines.
\iflncs
We omit the simple proofs that the problems are in $\Sigma_1^0$ and $\Sigma_2^0$.
\else
We then report the simpler corresponding completeness proofs.
\fi

\textbf{Hardness of the bounded counter problem.}
Given a generic $n$-counter machine $M$, we reduce halting to bounded counter for $\beta = 0$ by modifying $M$ into $M'$ as follows.
Add one counter and injectively rename all counters in the instruction list so that the new counter is called $v_0$; thus, $v_0$ is not mentioned in the renamed instructions.
Then, replace the unique halting instruction appearing at some $\ell_h$ in $M$ by two instructions: \mbox{\lstinline|$\ell_h$: inc $\;v_0$|} followed by \lstinline|$\ell_h^+$: halt|.

Since we only added deterministic instructions, there is a one-to-one correspondence between computations of $M$ and computations of $M'$.
A generic nondeterministic computation $\chi$ of $M$ reaches location $\ell_h$ iff the unique corresponding computation $\chi'$ of $M'$ also reaches $\ell_h$.
In such computations $\chi'$, $v_0$ overflows $\beta$ before halting at $\ell_h^+$.
In all, some computation of $M$ halts iff $v_0$ overflows in some computation of $M'$.
Thus, the bounded counter problem is $\Sigma_1^0$-hard.

\textbf{Hardness of the finite counter problem.}
Given a generic $n$-counter machine $M$, we reduce from the non-halting problem.
Create another counter machine $M'$ with a fresh counter $v_0$, which works as follows.
$M'$ simulates all computations of $M$ deterministically: as soon as a specific computation terminates, $M'$ backtracks the simulation and makes a different nondeterministic choice.
(We omit the details of the simulation, which are straightforward.)
Whenever the simulation completes a halting computation of $M$, it increments $v_0$ before continuing with the next computation.
If the simulation ever comes to an end (that is, if $M$ has only finitely many computations, all halting), $M'$ enters an infinite loop that makes $v_0$ diverge.
Therefore, $M'$ has only one non-halting (because either $M'$ enters the infinite loop or $M$ has infinitely many computations) deterministic execution.

Consider now the finite counter problem for $M'$.
If it has answer \textsc{yes}, it means that the simulation eventually executes a non-halting computation of $M$; from that point on, $v_0$ is never incremented.
If it has answer \textsc{no}, it means that the simulation consists of infinitely many halting computations of $M$, or that it reached the divergent loop and hence $M$ had only finitely many halting computations.
The answer to the non-halting problem for $M$ is therefore the same in either case.
This shows that we reduced the non-halting problem to the finite counter problem, and both are $\Sigma_2^0$-hard.
\iflncs\end{proof}\fi

\exOnly{
\textbf{Hardness of the unbounded counter problem.}
Given a generic $n$-counter machine $M$, we reduce recurring computation to unbounded counter by modifying $M$ into $M'$ as follows.
Add one counters and injectively rename all counters in the instruction list so that the new counter is called $v_0$. 
Then, replace the instruction $I$ appearing at location $\ell_0$ in $M$ by two instructions as follows: \lstinline|$\ell_0$: inc $\:v_0$; $\;$ $\ell_0^\prime$: $\:I$|. All other instructions follow $\ell_0'$ as they followed $\ell_0$ in $M$.

Also in this case we only added deterministic instructions; hence there is a one-two-one correspondence between computations of $M$ and computations of $M'$.
A generic nondeterministic computation $\chi$ of $M$ visits location $\ell_0$ infinitely often iff the unique corresponding computation $\chi'$ of $M'$ also reaches the new $\ell_0$ infinitely often; such computations $\chi'$ increment $v_0$ infinitely often when executing $\ell_0$.
In all, some computation of $M$ visits $\ell_0$ infinitely often iff $v_0$ is incremented infinitely often in some computation of $M'$.
Thus, the unbounded counter problem is $\Sigma_1^1$-hard.
}

\iflncs\else
\textbf{Completeness of the bounded counter problem.} 
We reduce the \ifplain\linebreak\fi bounded counter problem (for any $\beta$) to halting, thus showing that the former is in $\Sigma_1^0$ (and hence, by combining it with the hardness result, $\Sigma_1^0$-complete).
The idea is to guard every increment to $v_0$ with a conditional of the form 
\mbox{\lstinline|if  $\;v_0 \geq \beta\;$ goto $\;\ell_h\;$ else inc $\;v_0$|,} where $\ell_h$ is the halting location.
Since $v_0$ is initially zero, a computation halts iff it overflowed in the initial program.
The details of how to encode such modifications using standard instructions are straightforward.

\textbf{Completeness of the finite counter problem.} 
We show that the finite counter problem is in $\Sigma_2^0$ (and hence, by combining it with the hardness result, $\Sigma_2^0$-complete) according to the definition of $\Sigma_2^0$ in the arithmetical hierarchy~\cite{Rog87}.
Let $\mathcal{O_\beta}$ be the set of all counter machines where $v_0$ overflows $\beta$ in some computation.
Previously, we have shown that $\mathcal{O_\beta}$ is $\Sigma_1^0$; hence its complement set $\overline{\mathcal{O_\beta}}$---all counter machines where $v_0 \leq \beta$ in all computations---is $\Pi_1^0$.
The set $\mathcal{F}$ of all counter machines for which the finite counter problem has answer \textsc{yes} is defined by $M \in \mathcal{F} \Longleftrightarrow \exists \beta: \overline{\mathcal{O_\beta}}$, and hence it is $\Sigma_2^0$.
\end{proof}
\fi


\fi

\subsection{MTL and Counter Machines}\label{sec:mtl-counter-machines}
Alur and Henzinger~\cite{AH93} pioneered the usage of counter machines to analyze the complexity of real-time logics.
Using their techniques, we show the essentials of how to encode computations of $n$-counter machines as MTL formulas over $\reals_{\geq 0}$: computations are encoded as timed words; and, given a machine $M$, we build an MTL formula $\Gamma_M$ that is satisfied precisely by the words encoding $M$'s computations.

Consider an $n$-counter machine $M$ with $m+1$ instructions $\ell_0, \ldots, \ell_m$, such that $\ell_h$ is the location of the unique halt instruction.
We introduce the following propositions: $p_k$, for $0 \leq k \leq m$, which holds when $M$ is at location $\ell_k$; and $z_k$, for $1 \leq k \leq n$, which we use to represent the value of counter $v_k$: there are as many distinct occurrences of proposition $z_k$ over a unit interval as the value of counter $v_k$ in the corresponding configuration.
A configuration is a tuple $\langle \ell_k, x_1, \ldots, x_n \rangle$ denoting that $M$ is at location $\ell_k$ and the counters store the values $x_1, \ldots, x_n$.
At each integer time instant: all propositions $z_d$'s are false; and exactly one of the propositions $p_k$'s holds, with $p_0$ holding initially. The $p_k$'s are all false everywhere else:
\[
p_0 
 \land\!
\left(\ifieee\!\!\!\fi
\begin{array}{l}\!
\bigwedge_{1 \leq k \leq m} \!\boxMTL{}{p_k \limpl \bigwedge_{1 \leq j \neq k \leq m} \neg p_j \,\land\, \bigwedge_{1 \leq d \leq n} \neg z_d} \!\,\land \\
\!\bigwedge_{1 \leq k \leq m} \boxMTL{}{p_k \limpl \bigvee_{1 \leq j \leq m} \untilMTL{= 1}{\bigwedge_{1 \leq i \leq m} \neg p_i, p_j}\!}
\end{array}\ifieee\!\!\!\fi
\!\right)\!.
\]
With similar formulas, we constrain the $z_k$'s to occur at distinct instants: whenever $z_k$ then $\neg z_h$ also holds simultaneously, for $h \neq k$.

Each time interval $[t, t+1)$, for $t \in \naturals$, encodes the $(t+1)$-th configuration reached during a valid computation: $p_k$ holding at $t$ means that $M$ is at location $\ell_k$; and, for $1 \leq j \leq n$, $z_j$ holds over $[t, t+1)$ exactly as many times as the integer value stored in counter $v_j$.
The initial configuration $\langle \ell_0, 0, \ldots, 0 \rangle$ is encoded by
\[
\bigwedge_{1 \leq j \leq n} \boxMTL{[0,1]}{\neg z_j}\,.
\]

The encoding of any instruction refers to a current time $t \in \naturals$ and defines the state over $[t+1, t+2)$ as a modification of the state over $[t, t+1)$.
The most significant operation is the increment: \lstinline|$\ell_k$: inc $v_c$|, whose MTL encoding declares that the state in the next interval has exactly one more occurrence of $z_c$ than it has in the current interval:
\begin{equation}
\boxMTL{}{}
\left(
p_k \limpl
\left(
\begin{array}{l}
\diamondMTL{=1}{}\,p_{k+1}  \\
\land\ \bigwedge_{1 \leq d \neq c \leq n} \boxMTL{(0, 1)}{z_d \liff \diamondMTL{=1}{}\,z_d} \\
\land\ \;\boxMTL{(0, 1)}{z_c \limpl \diamondMTL{=1}{}\,z_c} \\
\land \,\untilMTL{(0, 1)}{}
   \left( \begin{array}{l}
   \diamondMTL{=1}{}\,z_c \limpl z_c, \\
   \neg z_c \land \diamondMTL{=1}{}z_c \ \land \\
   \untilMTL{>0}{\neg z_c \land \diamondMTL{=1}{\neg z_c}, p_{k+1}}
 \end{array}\!\right)
\end{array}\ifieee\!\!\fi
\!\right)\ifieee\!\!\fi
\!\right)\ifieee\!\!\fi.
\label{eq:inc-mtl-R}
\end{equation}
In \eqref{eq:inc-mtl-R}'s consequent, the first conjunct states that $\ell_{k+1}$ is the next location visited (since this is not a branch instruction).
The second conjunct states that the values of all counters other than $v_c$ are unchanged: for every occurrence of some $z_d$ in the current interval, there is an occurrence exactly one time unit later in the next interval and vice versa; hence occurrences of $z_d$ are ``copied'' from the current to the next interval.
Similarly, the third conjunct declares that $v_c$  does not decrease ($z_c$'s occurrences in the current interval are copied into the next one).
The fourth conjunct asserts that there exists an instant, after the last occurrence of $z_c$ in the current interval and before the next occurrence of $p_{k+1}$ at the beginning of the next interval, such that $z_c$ occurs exactly once at the corresponding instant in the next interval.
This new distinct occurrence of $z_c$ is always possible thanks to the density of the temporal domain; thus any value of counters can be stored in a unit time interval.
The encoding of other instructions is similar, with the halting instruction determining an indefinite repetition of the final configuration in the future.

\begin{remark}
Over pointwise semantics, we can express a behavior analogous to \eqref{eq:inc-mtl-R} using past operators.
The key observation~\cite{DBLP:journals/lmcs/OuaknineW07} is that the ``copy'' of a counter $v_d$ can be expressed as $\boxMTL{(0, 1)}{z_d \limpl \diamondMTL{=1}{}\,z_d}$ and $\boxMTL{(1, 2)}{}(z_d \limpl \diamondPMTL{=1}{}\,z_d)$, where $\diamondPMTL{=1}{\phi}$ holds iff its arguments held one time unit in the past.
\end{remark}

\section{The Complexity of Bounded Variability} \label{sec:compl-bound-vari}

Given a time domain $\timedomain$ and a formula $\phi$, we define two decision problems---the second is a generalization of the first---that deal with $\phi$'s bounded variability.
We write $\mathcal{B}(\phi)$ to denote the subset of a set $\mathcal{B}$ of timed words that satisfy $\phi$. 

\ifieee
\begin{description}[\IEEEsetlabelwidth{\BVt{v, V}: }]
\else
\begin{description}
\fi
\item[\BVt{v,V}:] Determine whether every model of $\phi$ over $\timedomain$ has variability bounded by $v/V$: does $\alltw{\phi} \subseteq \bwt{v/V}{\phi}$?

\item[\BVt{}:] Determine whether there exist $v, V$ such that the answer to \BVt{v,V} is \textsc{yes}: does 
$\exists\,v, V: \alltw{\phi} \subseteq \bwt{v/V}{\phi}$?
\end{description}
\noindent
A bar denotes the corresponding \emph{complement} problems: \cBVt{v, V} asks whether some model of $\phi$ has variability not bounded by $v/V$ (bounded by $v'/V$ for some $v' > v$, or unbounded); \cBVt{} asks whether, for every $v, V$, some model of $\phi$ has variability not bounded by $v/V$.
Notice that the latter is not the same as asking if some model of $\phi$ has unbounded variability: it may as well be that every model of $\phi$ has bounded variability, but no variability bounds all of the models.

This section establishes the complexity of the decision problems for $\timedomain = \reals_{\geq 0}$ (Section~\ref{sec:compl-bound-vari-ct}) and $\timedomain = \naturals$ (Section~\ref{sec:compl-bound-vari-dt}).

\subsection{Complexity of Bounded Variability over Continuous Time}\label{sec:compl-bound-vari-ct}
Both variants of the bounded variability problems just introduced are undecidable over continuous time, but with different undecidability degrees in the arithmetical hierarchy; in both cases, however, the undecidability degree is lesser than MTL satisfiability, which is highly undecidable ($\Sigma_1^1$-hard~\cite{AH93}).

\begin{theorem} \label{th:c-bounded-all}
\BVc{v, V} is $\Pi_1^0=\cRE$-complete; \BVc{} is $\Sigma_2^0$-complete.
\end{theorem}
\ifplain
\begin{proof}
The completeness result for $\BVc{v, V}$ is proved in Lemmas~\ref{lm:c-bounded-in} and~\ref{lm:c-bounded-hard}.
The completeness result for $\BVc{}$ is proved in Lemmas~\ref{lm:c-unbounded-in} and~\ref{lm:c-unbounded-hard}.
\iflncs\qed\else\qedhere\fi\end{proof}
\fi

\begin{lemma}\label{lm:c-bounded-in}
\BVc{v, V} is in $\Pi_1^0=\cRE$.
\end{lemma}
\begin{proof}
We give a procedure to semi-decide \cBVc{v,V}; this establishes that $\cBVc{v,V} \in \RE$ and thus $\BVc{v,V} \in \cRE$ by complement.

Consider a generic MTL formula $\phi$.
Some model of $\phi$ has variability not bounded by $v/V$ iff: (a) some model of $\phi$ has variability bounded by $v'/V$ but not by $v/V$, for some $v' > v$; or (b) some model of $\phi$ has unbounded variability.
Since we are dealing with divergent models only (see Section~\ref{sec:words}), (b) can only occur with models where the variability is bounded up to any finite time $t$, but the variability bound increases indefinitely over time.\footnote{In related work, we called similar behaviors ``Berkeley''~\cite{FR10-TOCL10,FMMR-TimeBook-12}.}

For any finite time $T$, let $\phi[T]$ denote the MTL formula which restricts the evaluation of $\phi$ to the finite time interval $[0, T]$.
This can be constructed as follows: add a fresh proposition $e$ constrained by $\phi_e = \untilMTL{=T}{e, e \land \boxMTL{> 0}{}\neg e}$.
Rewrite $\phi$ in negation normal form, and replace every atom $q$ by $e \limpl q$. 
Postulate that, if $e$ is false, all other propositions in $\Pcal$ are false as well: $\phi_{\Pcal} = \boxMTL{}{}(\neg e \limpl \bigwedge_{p \in \Pcal} \neg p)$.
Finally, $\phi[T]$ is $\phi_e \land \phi \land \phi_{\Pcal}$.
Since no event occurs after finite time $T$, all models of $\phi[T]$ have variability bounded by $x/T$, for some finite (possibly very large) $x$.
Therefore, some model of $\phi[T]$ has variability not bounded by $t/T$ iff some model of $\phi[T]$ has variability bounded by $t'/T$ but not by $t/T$, for some $x \geq t' > t$.

We can now describe a procedure $P_1$ that semi-decides \cBVc{v,V}; it consists of the following steps:
\begin{stdenum}
\item Initially, $\delta := v + 1$ and $\Delta := V + 1$;
\item \label{step2} Using Corollary~\ref{prop:MTL-bounded-sat}, decide whether $\phi[\Delta]$ has some model with variability bounded by $\delta/V$ but not by $v/V$;
\item If it does, stop and return \textsc{yes};
\item Otherwise $\delta := \delta + 1$, $\Delta := \Delta + 1$, and go to (\textit{\ref{step2}}).
\end{stdenum}
If the answer to \cBVc{v, V} is \textsc{yes}, then either (a) or (b) above holds; let us show that, in both cases, $P_1$ terminates with the correct answer.

If (a) is the case, let $\omega_a$ be a model with variability bounded by $v'/V$ but not by $v/V$ for some $v' > v$; that is, $\omega_a$ has $\overline{v}$ events, for $v < \overline{v} \leq v'$, over some time interval $[x, x+V]$.
In this case, $P_1$ terminates with \textsc{yes} as soon as $\delta \geq \overline{v}$ and $\Delta \geq x + V$.

If (b) is the case, let $\omega_b$ be a model with unbounded variability; since variability is unbounded, there exists a time $T$ such that: $\omega_b$ has $v' > v$ events over some time window $[x, x+V]$, for $0 \leq x < x+V \leq T$.
In this case, $P_1$ terminates with \textsc{yes} as soon as $\delta \geq v'$ and $\Delta \geq T$.
\iflncs\else\qedhere\fi\end{proof}

\begin{lemma}\label{lm:c-bounded-hard}
\BVc{v, V} is $\cRE$-hard.
\end{lemma}
\begin{proof}
We reduce the bounded counter problem (Section~\ref{sec:bound-unbo-count}) of $2$-counter machines to \cBVc{v, V}; the lemma follows by Theorem~\ref{th:counter-problems} through complement problems.

Consider a generic $2$-counter machine $M$ with counters $v_0$ and $v_1$.
We construct an MTL formula $\Gamma_{M}$ that encodes the computations of $M$ along the lines of Section~\ref{sec:mtl-counter-machines}, but with some modifications.
For $t \in \naturals$, the $(t + 1)$-th configuration $\langle \ell_k, x_0, x_1 \rangle$ is encoded over the time interval $[4t, 4t + 4)$ as follows: $p_k$ holds at $4t$, $z_0$ holds $x_0$ times over $(4t + 1, 4t + 2)$, $z_1$ holds $x_1$ times over $(4t + 3, 4t + 4)$, and no propositions hold elsewhere over the whole $[4t, 4t + 4)$.
With this spacing of counter events, we can see that the models of $\Gamma_{M}$ are such that any interval of length $1$ includes at most as many events as the largest value held by a counter during some computation.
Thus, $\Gamma_{M}$ has some model with variability not bounded by $\beta/1$, which is an instance of  \cBVc{v, V}, iff a counter overflows $\beta$ in some computation of $M$.

Now we have only established whether \emph{some} counter overflows in $M$, whereas the bounded counter problem specifically targets overflows of $v_0$.
To close the gap, we encode the overflowing of $v_0$ in $M$ as an MTL formula $\Xi_\beta^{v_0}$:
\[
\diamondMTL{}{}\left(\! \left(\bigvee_{0 \leq k \leq m}p_k \right) \land \overbrace{\diamondMTL{(0, 1)}{z_0 \land \diamondMTL{(0, 1)}{z_0 \land \cdots}}}^{\beta+1 \text{ nested diamonds}} \right) \,.
\]
Thanks to the padding, the nested diamonds evaluate to true iff there are at least $\beta+1$ distinct occurrences of $z_0$ in the slot corresponding to one configuration.
Thus, $v_0$ overflows $\beta$ in $M$ iff $\Gamma_{M} \land \Xi_\beta^{v_0}$ has some model with variability not bounded by $\beta/1$.
\iflncs\else\qedhere\fi\end{proof}

\begin{lemma} \label{lm:c-unbounded-in}
\BVc{} is in $\Sigma_2^0$.
\end{lemma}
\begin{proof}
Given the definition of $\Sigma_2^0$ in the arithmetical hierarchy~\cite{Rog87}, it is sufficient to provide an enumeration of all MTL formulas $\phi$ for which the answer to \BVc{} is \textsc{yes}, relative to an oracle for \BVc{v, V}, which is in $\Pi_1^0$ by Lemma~\ref{lm:c-bounded-in}.
To this end, we dovetail~\cite[Chap.~3]{Pap94} through all pairs $(v, \phi)$ of nonnegative integers $v \in \naturals$ and MTL formulas $\phi$.
For each pair, if the answer to \BVc{v, 1} is \textsc{yes} for $\phi$, then the answer to \BVc{} also is \textsc{yes} for $\phi$.
It is clear that this enumeration eventually finds all formulas for which the answer to \BVc{} is \textsc{yes}.
\iflncs\else\qedhere\fi\end{proof}

\begin{lemma} \label{lm:c-unbounded-hard}
\BVc{} is $\Sigma_2^0$-hard.
\end{lemma}
\begin{proof}
We reduce the finite counter problem (Section~\ref{sec:bound-unbo-count}) of $n$-counter machines to \BVc{}; the lemma follows by Theorem~\ref{th:counter-problems}.

This reduction is the trickiest among those in this paper. 
The difficulty lies in the fact that, while the finite counter problem refers a specific counter $v_0$, \BVc{} considers variability of all propositions; while it is easy to reduce from general to specific, here we need to build a reduction in the opposite direction.
The proof of Lemma~\ref{lm:c-bounded-hard} involves a similar mismatch, but things are simpler there, thanks to the existence of a known bound $\beta$, which we can monitor explicitly; now, instead, the bound is existentially quantified.
An easy solution would be to change the definition of \BVc{} to refer a specific proposition that varies, but that would weaken the result proved.
Instead, we leverage nondeterminism to ``guess'' the bound.

Build a counter machine $M_x$ which simulates computations of $M$ as follows.
Every computation of $M_x$ starts by nondeterministically storing a positive integer $x$ in a fresh counter $v_x$.
This is achieved by the instructions:
\begin{equation}\label{eq:nondet-shuffle}
\begin{split}
  &\ell_0: \text{ \lstinline|inc| }v_x \\
  &\ell_1: \text{ \lstinline|if| } v_{x} > 0 \text{ \lstinline|goto| } \ell_0, \ell_2
\end{split}
\end{equation}
If $\ell_1$'s nondeterministic branch eventually jumps to $\ell_2$, the rest of $M_x$'s program simulates all computations of $M$ by dovetailing~\cite[Chap.~3]{Pap94}, so that the simulation does not get stuck in non-terminating computations of $M$.
Additionally, whenever $v_0$ overflows $x-1$, the simulation halts; and if $M$ had only finitely many computations, the simulation concludes with an infinite idle loop (unless it has previously halted upon $v_0$ overflowing).

Consider now the MTL formula $\Gamma = \Gamma_{M_x} \land \diamondMTL{}{}\,p_{h}$, where $\Gamma_{M_x}$ encodes $M_x$'s computations as in Section~\ref{sec:mtl-counter-machines}, and $\ell_h$ is the unique halting location of $M_x$.
Thus, the models of $\Gamma$ describe all valid computations of $M_x$ that halt (and hence, in particular, that do not get stuck forever in the initial loop \eqref{eq:nondet-shuffle} that increments $v_x$---something we get for free given that we are reducing between undecidable problems).

Consider now \BVc{} for $\Gamma$.
If its answer is \textsc{yes}, then there must be only finitely many models that satisfy $\Gamma$; otherwise, they would include simulations for all values of $x$, which would entail that, for every possible bound $x$, there exists a model where $v_0$ overflows $x$, against the hypothesis that the answer is \textsc{yes} (i.e., all models are bounded).
Therefore, there is a finite bound on $v_0$ in all computations of $M$, given by one plus the maximum of values reached by $v_0$ in all finitely many models.
Conversely, if the answer to \BVc{} for $\Gamma$ is \textsc{no}, then there must be infinitely many models that satisfy $\Gamma$, that is one for every value of $x$; in fact, all such models are halting, and hence if they are finitely many the maximum of all counters in all such models would be well defined and finite, against the hypothesis that the answer is \textsc{no} (i.e., there always is an unbounded model).
The existence of halting models for all values of $x$ entails that $v_0$ overflows any finite value in some computation.
In summary, the answer to \BVc{} for $\Gamma$ is \textsc{yes} iff the answer to the finite counter problem for $M$ is \textsc{yes}.
This concludes the reduction.
\iflncs\else\qedhere\fi\end{proof}

\subsection{Complexity of Bounded Variability over Discrete Time}\label{sec:compl-bound-vari-dt}
The complexity of bounded variability over discrete time paints a picture quite different from that over continuous time.
Lemmas~\ref{lm:d-bounded-hard}--\ref{lm:d-bounded-in} prove that \BVd{v, V} is $\expspace$-complete, which is the same complexity as MTL satisfiability over $\naturals$.
On the other hand, it is clear that \BVd{} is decidable in constant time, since every discrete-time MTL formula has variability bounded by $x/x$ for any integer $x > 0$, precisely because the time domain is discrete, and hence there is a hard upper bound on the variability of events.

\begin{lemma}\label{lm:d-bounded-hard}
\BVd{v, V} is $\expspace$-hard.
\end{lemma}
\begin{proof}
We polynomial-time reduce MTL satisfiability to \cBVd{v, V}; the lemma follows since $\expspace$ is closed under complement.

The decision procedure for discrete-time TPTL is based on the following fundamental property~\cite[Lemma~5]{AH94}: a TPTL formula $\psi$ is satisfiable iff it has a model where the difference between any pair of consecutive timestamps is always less than or equal to the product $\delta_{\psi}$ of all constants appearing in $\psi$.
The same property holds of MTL formulas over the integers~\cite{AH93}. 

Since we are considering timed $\omega$-words, which have infinitely many events, the property entails that an MTL formula $\phi$ is satisfiable iff there exists a timed word $\omega$ such that: $\omega$ has at least one event with timestamp $t \leq \delta_{\phi}$ and $\omega \models \phi$.
Therefore, a generic MTL formula $\phi$ is satisfiable iff some of its models have variability not bounded by $0/\delta_{\phi}$, that is iff the answer to $\cBVd{0, \delta_{\phi}}$ is \textsc{yes}.
Assuming a binary encoding of constants, as it is customary, $\delta_{\phi}$ is polynomial in the size of $\phi$ (because the product of $n$ constants has size $\Orm(n^2)$ in binary), thus the reduction is done in polynomial time.
\iflncs\else\qedhere\fi\end{proof}

\begin{lemma}\label{lm:d-bounded-in}
\BVd{v, V} is in $\expspace$.
\end{lemma}
\begin{proof}
We show how to encode the requirement that a model has variability bounded by $v/V$ as an MTL formula $B_{v,V}$.

If $v = 0$, then $B_{v, V} = \boxMTL{}{}\logfalse$.
Otherwise, we can adapt the techniques we introduced for LTL~\cite{FS11-TIME11}.
Consider $v > 0$ fresh propositions $p_i$, for $i = 1, \ldots, v$.
Proposition $p_1$ holds initially, followed by $p_2, \ldots, p_v$ in sequence; the sequence repeats indefinitely:
\[
B_v \ifieee\else\ \fi=\ifieee\else\ \fi p_1 \land \bigwedge_{1 \leq k \leq v} \!\!\!\left(\!\boxMTL{}{p_1 \liff \nowonMTL{}{}\,p_{k \oplus 1}} \land \boxMTL{}{p_k \limpl \!\!\!\!\bigwedge_{1 \leq h \neq k \leq v} \neg p_h} \!\!\!\right)
\]
where $a \oplus b$ is a shorthand for $1 + ((a+b) \bmod v)$.
Since every $p_k$ holds in a different position, we can express bounded variability by requiring that the timestamp of the next $(v+1)$-th position in the future be greater than $V$ with respect to the current position's (and note that $k \oplus v = k$):
\[
B_{v, V} \quad =\quad B_v \land \bigwedge_{1 \leq k \leq v} \boxMTL{}{p_k \limpl \untilMTL{> V}{\neg p_k, p_{k}}}\,.
\]
Thus, $\phi \limpl B_{v, V}$ is valid iff the answer to \BVd{v, V} for $\phi$ is \textsc{yes}.

The only problem with this reduction is that $B_{v, V}$ has size exponential in the size of the instance of \BVd{v, V} assuming a binary encoding of constants. 
Precisely, the blow-up occurs because $B_v$ has size polynomial in $v$, which is exponential in the size of a binary encoding of $v$.
Encoding the modulo-$v$ counter in binary (using $n = \lfloor \log_2 v \rfloor + 1$ propositions) would not help: while updates to the counter itself can be done with formulas of size polynomial in $n$, there is no easy way to express in MTL the fact that the timestamp of the ``next'' occurrence is greater than $V$ (with respect to the current position's) without enumerating all $2^n = v$ values for the counter.

\exOnly{
Let us illustrates the problem, assuming for simplicity, but without loss of generality, that $v = 2^n$ for some integer $n$.
Consider $n$ propositions $b_1, \ldots, b_n$ such that a $b_k$ represent the $k$-th bit of a counter spanning the $2^n$ values from $0^n$ to $1^n$; and $b_n$ is the most significant bit.
To simplify the notation, we write $\neg b_k$ as $\ob_k$, and string such as $b_n \cdots b_1$ represent propositional formulas such as $b_n \land \cdots \land b_1$.
From one position to the next, the counter gets incremented by one.
In binary, this is expressed as follows: starting from the least significant bit, flip all $1$s until you reach the first $0$; flip the $0$ as well, and leave all other more significant bits unchanged:
\[
\bigwedge_{1 \leq k \leq n} \!\!\!
\left(
\ob_k b_{k-1} \cdots b_1 
\limpl
\nextMTL{}{b_k \ob_{k-1} \cdots \ob_1}
\land\!\!\!\!\!
\bigwedge_{k < j \leq n} \!\!(b_j \liff \nextMTL{}{}\,b_j)
\!\!\right)
\]
plus the special case $b_n \cdots b_1 \limpl \ob_n \cdots \ob_1$ specified separately.
This formulas has size $\Orm(n^2)$, but expressing bounded variability also requires a formula:
\[
\boxMTL{}{x_n \cdots x_1 \limpl \boxMTL{> V}{x_n \cdots x_1}}
\]
for each of the $2^n$ values $x_n \cdots x_1$ of the bits $b_1, \ldots, b_n$.
}

The blow-up is, however, inessential and only due to the fact that MTL operators do not include compact ``counting'' modalities.
We omit the details for brevity, but it is clear that one can extend the standard decision procedures for MTL~\cite{AH93} to handle counting modalities without affecting the complexity of the logic.
Specifically, we could introduce an operator $\CntMTL{J}{n}\psi$ with the semantics: $\omega, k \models \CntMTL{J}{n}\psi$ iff $t_{k+n} - t_k \in J$ and $\omega, k+n \models \psi$.
$B_{v, V}$ for $v > 0$ is then equivalent to $\boxMTL{}{\CntMTL{> V}{v}\logtrue}$; assuming a binary encoding of constants, this has size linear in the size of the encodings of $v$ and $V$.
\iflncs\else\qedhere\fi\end{proof}

\section{Bounded Variability in Simple Cases} \label{sec:bound-vari-simple}
The complexity results of Section~\ref{sec:compl-bound-vari} pose some major limitations to deciding bounded variability for \emph{generic} MTL formulas.
However, the outlook may be better if we target \emph{fragments} of MTL that are still sufficiently expressive but for which reasoning about bounded variability is simpler than in the general case.
We call such fragments ``bounded friendly''.
We give two examples of non-trivial bounded-friendly fragments, one for discrete and one for dense time.

\ifplain
\begin{definition}\label{def:bound-friendly}
An MTL fragment $\Fcal$ is \emph{bounded friendly} over $\timedomain$ iff three conditions hold:
\begin{enumerate}
\item \label{req1} We can express in $\Fcal$ a sufficient condition for bounded variability; that is, for any $v, V$, there exists a computable formula $B_{v,V} \in \Fcal$ such that all models of $B_{v,V}$ have variability bounded by $v/V$.
\newcounter{latestEnumCnt}
\setcounter{latestEnumCnt}{\theenumi}
\end{enumerate}
Given a generic MTL formula $\phi$, we can construct two formulae $\psi$ and $\phi'$ such that:\footnote{For consistency, assume all complexities are time complexities.}
\begin{enumerate}
\setcounter{enumi}{\thelatestEnumCnt}
\item \label{req1bis} $\phi$ is satisfiable iff $\psi \land \phi'$ is.
\item \label{req2} There exists a formula $\psi' \in \Fcal$ equivalent to $\psi \limpl B_{v,V}$ and constructable in $\Orm(b(|\phi|))$.
\item \label{req3} Deciding validity $\psi \limpl B_{v,V}$ is simpler than deciding bounded variability for $\phi$; that is, if validity for $\gamma \in \Fcal$ is decidable in $\Orm(f(|\gamma|))$, and $\psi \limpl B_{v,V}$ is constructable in , then $f(b(x))$ is $\orm(m(x))$, where $m(x)$ bounds the complexity of deciding \BVt{v,V}.
\end{enumerate} 
\end{definition}

For a bounded-friendly MTL fragment, we can proceed as follows.
Rewrite $\phi$ into $\psi \land \phi'$; construct $\psi \limpl B_{v,V}$ and determine if it is valid; if it is, then all models of $\phi$ have variability bounded by $v/V$, since $\phi \limpl \psi$, but we determined it with less computational resources than by analyzing $\phi$ directly.
The challenge in making this process practical is finding sufficiently expressive fragments $\Fcal$, which can represent a ``large part'' of $\phi$, as well as bounded variability itself.
The following subsections discuss non-trivial MTL fragments that are also bounded friendly over the integers (Section~\ref{sec:simpler-d}) and over the reals (Section~\ref{sec:simpler-c}).
\fi

\subsection{Simpler Bounded Variability over Discrete Time}\label{sec:simpler-d}

Over discrete time, MTL essentially boils down to an exponentially succinct version of LTL.
Therefore, we can try to lift some complexity results about simpler fragments of LTL~\cite{DBLP:journals/jacm/SistlaC85,DBLP:journals/iandc/DemriS02,DBLP:journals/corr/abs-0812-4848} to MTL over $\naturals$, and use them to identify bounded-friendly fragments.

Consider the two dual MTL fragments \mtlFX{} and \mtlGX{}: \mtlFX{} (respectively, \mtlGX{}) denotes the MTL fragment using only the \diamondMTL{J}{} (respectively, \boxMTL{J}{}) and \nextMTL{J}{} modalities (which we now regard as primitive), the propositional connectives $\land$ and $\lor$, and where negations only appear on atomic propositions.
Satisfiability for these fragments is decidable in exponential \emph{time}.

\begin{lemma}\label{lem:complexity-d-FX+}
Satisfiability of \mtlFX{} and of \mtlGX{} over $\naturals$ is \exptime-complete.
\end{lemma}

\iflncs
\else
\begin{proof}\iflncs[of Lemma~\ref{lem:complexity-d-FX+}]$\:$\fi
Consider the LTL fragment \ltlFX{} which only uses the eventually and next LTL modalities, the propositional connectives $\land$ and $\lor$, and where negations only appear on atomic propositions; \cite[Th.~3.7]{DBLP:journals/jacm/SistlaC85} proves that satisfiability for \ltlFX{} is \np-complete.
We outline how to transform a generic $\mu \in \mtlFX{}$ into a $\lambda \in \ltlFX{}$ such that $\mu$ and $\lambda$ are equisatisfiable; the converse transformation is also derivable along the same lines.
In general, the size of $\lambda$ is exponential in the size of $\mu$ due to the fact that metric constraints are encoded in binary in $\mu$.
The lemma follows as a manifestation of the ``succinctness phenomenon'' \cite[Chap.~20]{Pap94}.

The models of \ltlFX{} are denumerable sequences $w = w_0\, w_1\cdots$ such that $w_k$ is the set of propositions that hold at step $k$.
We conventionally assume that a step in $w$ corresponds to one discrete time instant; thus, a generic $\naturals$-timed word $\omega = (\sigma_0, t_0)\,(\sigma_1, t_1)\cdots$ uniquely corresponds to a sequence $w = w_0\,w_1\cdots$ such that: for all $k \in \naturals$, $w_{t_k} = \sigma_{k}$; and, for all $h$'s that are not timestamps of $\omega$, $w_h = \{\noop\}$, where $\epsilon$ is a special proposition denoting absence of a reading.

We define a translation $\tau$ from $\mtlFX{}$ to $\ltlFX{}$ inductively as follows, for $a, b \in \naturals$ and $c \in \naturals \cup \{\infty\}$:
\[
\begin{array}{lcl}
\tau\left(\diamondMTL{[a, b]}{\pi}\right)
& = &
\Xmtl{a}(\pi_\tau \vee \overbrace{\Xmtl{}(\pi_\tau \vee \cdots)}^{b-a \text{ nested } \Xmtl{}s})
\\
\tau\left(\diamondMTL{[a, \infty)}{\pi}\right)
& = &
\Xmtl{a}\Fmtl{}(\pi_\tau)
\\
\tau\left(\nextMTL{[a, c]}{\pi}\right)
& = &
\overbrace{\Xmtl{}(\epsilon \land \Xmtl{}(\epsilon \land \cdots))}^{a-1 \text{ nested }\Xmtl{}s}
\:\land\:
\tau\!\left(\diamondMTL{[a,c]}{\pi}\right)
\end{array}
\]
where $\tau(\pi) = \pi_\tau$ and $\Xmtl{k}$ is a shorthand for $k$ nested applications of $\Xmtl{}$.
$\tau$ does not otherwise change the propositional structure of formulas.

It should be clear that a generic $\mu$ is satisfiable over timed words over $\naturals$ iff $\tau(\mu)$ is satisfiable: models of $\tau(\mu)$ are obtained from models of $\mu$ according to the mapping of timed words described above.
Furthermore, the size of $\tau(\mu)$ is $\Orm(2^{|\mu|})$, since $\tau$ ``unrolls'' the constants succinctly represented in $\mu$, which results in a possible exponential blow-up.
This establishes the lemma for \mtlFX{}.
The same complexity result for \mtlGX{} follows by duality of $\boxMTL{}{}$ and $\diamondMTL{}{}$.
\end{proof}


\fi

We can leverage Lemma~\ref{lem:complexity-d-FX+} to show that \mtlFX{} is \emph{bounded friendly}.
Let $v, V$ be any variability bounds, with $v > 0$ w.l.o.g.
First, note that the equivalent \mtlGX{} formulas $\boxMTL{}{}\boxMTL{(0, \nu]}{\logfalse}$ and $\boxMTL{}{}\nextMTL{> \nu}{\logtrue}$, for $\nu = \lceil V/v \rceil$, hold only for models with variability bounded by $v/V$ (specifically, they are stricter than the definition of bounded variability).
Consider now a generic MTL formula $\phi$ written as $\phi' \land \psi$, where $\psi \in \mtlFX{}$.
The implication $\psi \limpl B_{v, V} \equiv \neg \psi \vee B_{v, V}$, where $B_{v, V}$ is one of the two just defined \mtlGX{} formulas implying bounded variability, is an \mtlGX{} formula: push in the outermost negation $\neg \psi$, and use the duality between $\diamondMTL{}{}$ and $\boxMTL{}{}$.
The validity of $\psi \limpl B_{v, V}$ can thus be decided in singly exponential time (Lemma~\ref{lem:complexity-d-FX+}), as opposed to \BVd{v, V} or the validity of $\phi$ which are \expspace-complete: solving them for a generic $\phi$ takes \emph{time} doubly exponential in $|\phi|$.
We can thus decide whether $\psi \limpl B_{v, V}$ is valid in singly exponential time; if it is, $\phi$ has bounded variability a fortiori, and hence we can study its validity with the simplified algorithms~\cite{FS11-TIME11,FS12-TIME12}.

We can show by duality that \mtlFX{} is also bounded friendly; for example, instead of the validity of $\psi \limpl B_{v,V}$, we equivalently consider the unsatisfiability of $\psi \land \neg B_{v, V}$.





\subsection{Simpler Bounded Variability over Continuous Time}\label{sec:simpler-c}
While MTL is highly undecidable over dense time, a number of expressive yet decidable fragments thereof have been identified.
MITL is the fragment of MTL where intervals are non-punctual, that is non-singular; MITL is fully decidable with \expspace-complete complexity~\cite{AFH96,HR04}.
More recently, other decidable expressive fragments have been identified that allow singular intervals~\cite{DBLP:conf/formats/OuaknineW08}; BMTL and SMTL, in particular, are interesting because their expressive power is incomparable with MITL's.

From the point of view of deciding bounded variability, however, MITL remains the most suitable choice.
SMTL validity has a non-elementary decision problem; while this is still better than the undecidable \BVc{v, V}, it makes it intractable in practice.
BMTL validity, in contrast, is decidable in \expspace; however, BMTL cannot express invariance properties since only finite intervals are allowed, and it is clear that bounded variability is a form of invariance property since it involves whole timed words.

Let us show that MITL is bounded friendly.
Let $v, V$ be any variability bounds, with $v > 0$ w.l.o.g. (the limit case $v=0$ can be handled separately).
Note that the MITL formula $B_{v, V} = \boxMTL{}{}\boxMTL{(0, \nu]}{\logfalse}$, for $\nu = \lceil V/v \rceil$, subsumes variability bounded by $v/V$.
Consider now a generic MTL formula $\phi$ written as $\phi' \land \psi$, where $\psi$ is an MITL formula.
The implication $\psi \limpl B_{v, V}$ obviously also is an MITL formula.
The validity of $\psi \limpl B_{v, V}$ is thus decidable; if $\psi \limpl B_{v, V}$ is valid, $\phi$ has bounded variability a fortiori, and hence its validity is decidable~\cite{FR08-FORMATS08}.

As a final remark, notice how leveraging MITL's bounded friendliness can still be useful to determine the satisfiability of MTL specifications not entirely expressed in MITL: as long as the part $\psi$ expressible in MITL entails bounded variability, the rest $\phi'$ of the actual specification can use any MTL operator, including singular intervals.

\section{Discussion: Other MTL Semantics} \label{sec:other-semantics}
Remark~\ref{rem:semantics-used} clarified that the results of this paper assume infinite timed words, and the continuous semantics over dense time.
While these are perfectly common assumptions (and the naturalness of the pointwise semantics over dense time has been questioned~\cite{HR04}), it is still interesting to get an idea of how our results would change under a different semantics.

Over discrete time, it is straightforward to notice that all complexity results proved in \ref{sec:compl-bound-vari-dt} carry over to the finite-word semantics (where decidability also has the same \expspace\ complexity).
Over dense time, it is possible to extend the results of Section~\ref{sec:compl-bound-vari-ct} to the \emph{signal} semantics of \cite{AFH96,FR10-TOCL10}---which can be seen as a variant of the continuous semantics---by reusing some of the constructions and definitions of~\cite{FR08-FORMATS08}.

The situation over dense timed words (both finite and infinite) under the pointwise semantics (and no past operators) is different.
For both finite and infinite words, \BVc{v, V} is no more difficult than validity, because one can encode the bounded variability requirement as in Lemma~\ref{lm:d-bounded-in}.
Therefore, \BVc{v, V} is decidable (nonprimitive recursive) over finite words~\cite{DBLP:journals/lmcs/OuaknineW07}, and is $\RE$ over infinite words~\cite{OuaknineW06}.
One can prove matching lower bounds along the lines of the proof of Lemma~\ref{lm:d-bounded-hard}: 
the abstraction of clock valuations into clock regions used in the construction of~\cite{DBLP:journals/lmcs/OuaknineW07} is such that the time difference between any pair of consecutive timestamps in any word that satisfies a formula $\phi$ is bounded above by a finite constant $\delta_{\phi}$ that depends only on $\phi$.
Therefore, $\cBVc{0, \delta_\phi}$ for $\phi$ reduces to non-satisfiability of $\phi$, giving the matching lower bounds through complement.
Finally, \BVc{} is $\RE$ under the pointwise semantics: by definition of $\RE$ as existential quantification over a recursive relation (finite words); and by dovetailing through the possible values $v, V$ and enumerating $\BVc{v, V}$ for them (infinite words).
Finding matching lower bounds belongs to future work, which may exploit the connection between MTL over infinite words under the pointwise semantics and channel machines with insertion errors~\cite{OuaknineW06}.

\iflncs
\else
\section{Conclusions}
The strong negative results of the paper need not be the deathblow to leveraging bounded variability to simplify temporal reasoning.
From a broader perspective, we can still look at the glass as half-full: while deciding bounded variability is intractable in general, there are situations where the physical requirements of a system include a notion of \emph{finite speed}, which bounded variability naturally embodies.
For such systems, there is still hope of using the expressiveness of MTL without succumbing to the dark side of intractability.
\fi

\fakepar{Acknowledgements.}
Thanks to the anonymous reviewers of TIME 2014 for useful comments; in particular, reviewer~4's insightful observations helped simplify some of the proofs.

\ifplain
\fi
\ifieee
\fi
\iflncs
\fi


\end{document}